\documentclass[a4paper,UKenglish,cleveref, autoref]{lipics-v2019}

\usepackage{algorithm}
\usepackage{algorithmicx}
\usepackage{algpseudocode}
\usepackage{graphicx}
\usepackage{gastex}
\usepackage{tikz}
\usepackage{caption}
\usepackage{pict2e}

\algtext*{EndFunction}
\algtext*{EndProcedure}
\algtext*{EndFor}
\algtext*{EndWhile}
\algtext*{EndIf}
\algloopdefx{NIf}[1]{\textbf{if} #1 \textbf{then}}
\algloopdefx{NElse}{\textbf{else}}
\algloopdefx{NElseIf}{\textbf{else if}}
\algloopdefx{NForAll}[1]{\textbf{for each} #1 \textbf{do}}
\algloopdefx{NWhile}[1]{\textbf{while} #1 \textbf{do}}
\algloopdefx{NFor}[1]{\textbf{for} #1 \textbf{do}}
\nolinenumbers
\newcommand{\pl}{\mathsf{pl}}

\newcommand{\pre}{\mathsf{pre}}

\newcommand{\leftside}{\mathsf{left}}

\newcommand{\cntr}{\mathsf{cntr}}
\newcommand{\lst}{\mathsf{list}}
\newcommand{\len}{\mathsf{len}}
\newcommand{\append}{\mathsf{add}}
\newcommand{\maxPal}{\mathsf{maxPal}}
\newcommand{\rad}{\mathsf{rad}}
\newcommand{\nextPal}{\mathsf{nextPal}}
\newcommand{\lvec}[1]{\overset{{}_{\leftarrow}}{#1}}

\newcommand{\ttl}{\mathsf{ttl}}
\newcommand{\pttl}{\mathsf{pttl}}
\newcommand{\head}{\mathsf{head}}
\newcommand{\size}{\mathsf{size}}
\newcommand{\tail}{\mathsf{tail}}

\newcommand{\PL}{{\mathsf{PL}}}
\newcommand{\PRE}{{\mathsf{PRE}}}
\newcommand{\dead}{{\mathsf{dead}}}
\newcommand{\wait}{{\mathsf{wait}}}

\title{Palindromic k-Factorization in Pure Linear Time}

\author{Mikhail Rubinchik}{Ural Federal University, 
Ekaterinburg, Russia}{mikhail.rubinchik@gmail.com}{}{}
\author{Arseny M. Shur}{Ural Federal University, Ekaterinburg, Russia}{arseny.shur@urfu.ru}{}{}

\authorrunning{M. Rubinchik and A.\,M. Shur} 	

\Copyright{Mikhail Rubinchik and Arseny M. Shur}
\ccsdesc{Theory of computation~Design and analysis of algorithms}
\ccsdesc{Mathematics of computing~Combinatorial algorithms}

\keywords{stringology, palindrome, palindromic factorization, online algorithm}

\begin{document}

\maketitle

\begin{abstract}
Given a string $s$ of length $n$ over a general alphabet and an integer $k$, the problem is to decide whether $s$ is a concatenation of $k$ nonempty palindromes. Two previously known solutions for this problem work in time $O(kn)$ and $O(n\log n)$ respectively. Here we settle the complexity of this problem in the word-RAM model, presenting an $O(n)$-time online deciding algorithm. The algorithm simultaneously finds the minimum odd number of factors and the minimum even number of factors in a factorization of a string into nonempty palindromes. We also demonstrate how to get an explicit factorization of $s$ into $k$ palindromes with an $O(n)$-time offline postprocessing. 
\end{abstract}

\newpage
\section{Introduction}

\emph{Factorization}, or representation of a string as a concatenation of substrings, is a key tool in both algorithmic and combinatorial studies on strings. For example, both the Lempel--Ziv factorization \cite{LempelZiv} and the Lyndon factorization \cite{CFL58} are ubiquitous in combinatorics on words and stringology. Factorization into palindromes also  attracted the attention of researchers since 1970s. Recall that a string $s = a_1a_2\cdots a_n$ is a \emph{palindrome} if it is equal to its \emph{reversal} $\lvec{s} = a_n\cdots a_2a_1$; \emph{palindromic ($k$-)factorization} is a representation of a string as a concatenation of ($k$) nonempty palindromes. There is a bunch of results concerning palindromic factorization. One of them is a hardness result: deciding the existence of a factorization into distinct palindromes is NP-complete \cite{BGIKKPPS15}; for all other natural types of palindromic factorization, sooner or later  linear-time algorithms in the word-RAM model of computation were designed. Knuth, Morris, and Pratt \cite{KMP77} presented such an algorithm deciding whether a string has a factorization into even-length palindromes. Galil and Seiferas \cite{GaSe78} did the same for the factorization into palindromes of length ${>}1$, and also for 2-, 3-, and 4-factorization. The existence of $k$-factorization was shown to be decidable in $O(kn)$ time~\cite{KosolobovPalk}. Most of the recent results were related to \emph{palindromic length} of a string, which is the minimum number of factors in its palindromic factorization. There are several combinatorics papers, see e.g. \cite{FPZ13,Frid18,Saa17}, studying the conjecture that every aperiodic infinite string has finite factors of arbitrarily big palindromic length. On the algorithmic side, the palindromic length of a string of length $n$ was  shown to be computable in $O(n\log n)$ time~\cite{FiciCo,ISugimotoInenagaBannaiTakeda,RuSh18}.
In \cite{BKRS17}, the optimal $O(n)$ bound was reached, using bit compression and range operations.

A linear-time algorithm for palindromic length gave a hope for better results on palindromic $k$-factorization, because deciding the latter is equivalent to computing the minimum even and minimum odd number of factors in a palindromic factorization. An $O(n\log n)$ algorithm for palindromic length \cite[Prop.\,4.2]{RuSh18} can be transformed into an $O(n\log n)$ algorithm for even/odd palindromic length \cite[Prop.\,4.9]{RuSh18}. However, the properties of palindromic length used in the linear-time algorithm of \cite{BKRS17} do not hold for even/odd palindromic length. In this paper we show how to overcome the technical difficulties and present the following optimal result.  

\begin{theorem}\label{MainTheorem}
There exists an online algorithm deciding, in $O(n)$ time independent of $k$, whether
a length-$n$ input string over a general alphabet admits a palindromic $k$-factorization.
\end{theorem}

In addition, we show how to find a palindromic $k$-factorization explicitly within the linear time (but offline).
The paper is organized as follows: after preliminaries on strings and palindromes, we give an $O(n\log n)$-algorithm for computing the even/odd palindromic length of a string in Section 2. The main section is Section~3 where a linear-time algorithm is described. We conclude with an algorithm finding an explicit $k$-factorization (Section~4). Asterisks $\mathbf{(*)}$ indicate that proofs or additional details can be found in Appendix.

\subparagraph*{Preliminaries.}
For any $i,j$, $[i..j]$ denotes the range $\{k\in \mathbb{Z} \colon i \le k \le j\}$; we abbreviate $[i..i]$ as $[i]$. We use ranges, in particular, to index strings and arrays. 
For strings we write $s=s[1..n]$, where $n = |s|$ is the length of $s$. The \emph{empty string} is denoted by $\varepsilon$. A string $u$ is a \emph{substring} of $s$ if $u = s[i..j]$ for some $i$, $j$ ($j<i$ means $u=\varepsilon$). Such pair $(i,j)$ is not necessarily unique; $i$ specifies an \emph{occurrence} of $u$ at \emph{position} $i$. A substring $s[1..j]$ (resp., $s[i..n]$) is a \emph{prefix} (resp. \emph{suffix}) of $s$. An integer $p\in[1..n]$ is a \emph{period} of $s$ if $s[1..n{-}p] = s[p{+}1..n]$. We write $s^k$ for the concatenation of $k$ copies of $s$ (thus $s^k$ has period $|s|$).  We write $\lvec{s}$ for a string or array obtained from $s$ by reversing the order of elements (thus the equality $s=\lvec{s}$ defines a \emph{palindrome}). A substring (resp. suffix, prefix) that is a palindrome is called a \emph{subpalindrome} (resp. \emph{suffix-palindrome}, \emph{prefix-palindrome}). If $s[i..j]$ is a subpalindrome of $s$, the numbers $(j + i) / 2$ and $\lfloor (j - i + 1) /2 \rfloor$ are respectively the \emph{center} and the \emph{radius} of $s[i..j]$. Subpalindromes of odd (even) length have integer (resp., half-integer) centers. 

A ($k$-)\emph{pal-factorization} of $s$ is a representation $s=w_1\cdots w_k$, where $w_1,\ldots,w_k$ are palindromes. The minimum $k$ such that a $k$-pal-factorization of $s$ exists is the \emph{palindromic length} of $s$, denoted by $\pl(s)$. We also introduce the \emph{even palindromic length} $\pl^0(s)$ and \emph{odd palindromic length} $\pl^1(s)$ as the minimum even (resp., odd) $k$ among all such factorizations of $s$. If $s$ has no pal-factorization with even (odd) number of factors, we write  $\pl^0(s)=\infty$ (resp., $\pl^1(s)=\infty$). For example, $\pl(abcba)=\pl^1(abcba)=1$, $\pl^0(abcba)=\infty$; $\pl(acaaba)=\pl^0(acaaba)=2$, $\pl^1(acaaba)=5$.
A pal-factorization  $s=w_1\cdots w_k$ with $k\le |s|-2$ can be easily transformed into a $(k{+}2)$-pal-factorization: either $|w_i|\ge3$ for some $i$, so $w_i=aua$ for some letter $a$ and palindrome $u$, or $|w_i|=|w_j|=2$ for some $i,j$, so $w_i$ and $w_j$ can be replaced by four 1-letter factors. This leads to the following crucial observation.

\begin{lemma}[{\cite[Sect. 4.1]{RuSh18}}] \label{l:kvslen} 
(1) A string $s$ has a $k$-pal-factorization iff $\pl^{k \bmod 2}(s)\le k$.\\ (2) A $k$-pal-factorization of $s$ can be obtained in $O(|s|)$ time from its $\pl^{k \bmod 2}(s)$-pal-factorization.
\end{lemma}
\begin{remark} \label{r:oddeven}
Due to Lemma~\ref{l:kvslen}, throughout the paper we study the problem of existence of a $k$-pal-factorization  for a string $s$ as the problem of computing $\pl^0(s)$ and $\pl^1(s)$.
\end{remark}
Recall that the  method for computing $\pl(s)$ is dynamic programming: we compute the array $\pl[1..n]$ such that $\pl[i]=\pl(s[1..i])$ using an artificial initial value $\pl[0]=0$ and the rule $\pl[k]=1+\min_{i\in S_k}\pl[i{-}1]$, where $S_k$ is the set of positions of all suffix-palindromes of $s[1..k]$. This rule can be easily adapted to compute $\pl^0(s)$ and $\pl^1(s)$ (due to symmetry, we always write $\pl^j$ assuming $j\in\{0,1\}$). We define the arrays $\pl^j[1..n]$ and employ the scheme
\begin{equation}\label{e:dp}
\pl^0[0]=0,\ \pl^1[0]=\infty,\     \pl^j[k]=1+\min_{i\in S_k}\pl^{1-j}[i{-}1], \text{ where } S_k=\{i: s[i..k] = \overleftarrow{s[i..k]}\}
\end{equation}

Henceforth, $s$ is the input string of length $n$ and  $\PL[i]=(\pl^0[i],\pl^1[i])$.  All considered algorithms work in the unit-cost word-RAM model with $\Theta(\log n)$-bit machine words and standard operations like in the C language. Since the only operation used on alphabetic symbols is testing two symbols for equality, the algorithms are valid for the general unordered alphabet. All algorithms except Algorithm~\ref{a:fact} are online and work in \emph{iterations}: $i$th iteration begins with reading the symbol $s[i]$ and ends before reading $s[i{+}1]$.

\section{An $O(n\log n)$ Algorithm}
\label{s:nlogn}

\subparagraph*{Palindromic Iterator.}
We process the input string using a data structure called (palindromic) \emph{iterator} \cite{KosolobovPalk}, which stores a string $s$ and answers the following queries in $O(1)$ time:\\
- $\rad(x)$ / $\len(x)$ returns the radius / length of the longest palindrome in $s$ with the center $x$;\\
- $\maxPal$ returns the center of the longest suffix-palindrome of $s$;\\
- $\nextPal(x)$ returns the center of the longest proper suffix-palindrome of the suffix-palindrome of $s$ with the center $x$.\\
The iterator stores an array of radii for all possible centers of palindromes and a list of centers of suffix-palindromes in increasing order. The update query $\append(a)$ appends letter $a$ to $s$. Performing this query, the iterator emulates an iteration of Manacher's algorithm \cite{Manacher} and updates the list of suffix-palindromes, all within  $O(1+\maxPal_{new}-\maxPal_{old})$ time, which is $O(n)$ for $n$ updates to the originally empty structure. Below we assume that $\append(a)$ returns the list of deleted centers in the form $\dead(x)=(x, \text{answers to all queries about } x)$. We also write $\cntr(d)=n - (d-1)/2$ for the center of the length-$d$ suffix-palindrome of $s[1..n]$. 

With the iterator, the rule \eqref{e:dp} can be implemented to work in time proportional to the number of subpalindromes in $s$ ($\Omega(n^2)$ in the worst case); one iteration  is as follows:
\begin{algorithmic}[1]
\small
\State $\append(s[n]);\;\pl^j[n] \gets +\infty$
\For{($x \gets \maxPal;\; x \ne n + \frac{1}2;\; x\gets \nextPal(x)$)}
    \State $\pl^j[n] \gets \min\{\pl^j[n], 1 + \pl^{1-j}[n - \len(x)]\}$
\EndFor\vspace{-0.1cm}
\end{algorithmic}

\subparagraph*{Series of Palindromes.}
Let $u_1,\ldots, u_\ell$ be all non-empty suffix-palindromes of a string $s$ in the order of decreasing length. For any $i<j$, since $u_j$ is a suffix of $u_i$, any period of $u_i$ is a period of $u_j$. Hence the sequence of minimal periods of $u_1,\ldots, u_l$ is non-increasing. The groups of suffix-palindromes with the same minimal period are \emph{series of palindromes} (of $s$):
$$
\underbrace{u_1,\ldots,u_{i_1}}_{p_1},\underbrace{u_{i_1+1},\ldots,u_{i_2}}_{p_2},\ldots,\underbrace{u_{i_{t-1}+1},\ldots,u_\ell}_{p_t}.
$$
A \emph{$p$-series} is a series with the period $p$. The longest and the shortest palindrome in a $p$-series are its \emph{head} $\head(p)$ and \emph{tail} $\tail(p)$ respectively (they coincide in a 1-element series).  
\begin{lemma}[\cite{FiciCo,ISugimotoInenagaBannaiTakeda,KosolobovPalk}; $\mathbf{(*)}$] \label{l:lognseries}
For any string $s$, if $\ell$ is the length of a tail of a series, then the head of the next series has length less than $2\ell/3$. In particular, if $p_1>\ldots>p_t$ are periods of all series of $s$, then $t=O(\log p_1)$ and $p_1+\cdots+p_t=O(p_1)$.
\end{lemma}
Note that length-$n$ strings with $\Omega(\log n)$ series for $\Omega(n)$ prefixes do exist \cite{FiciCo}. The structure of series is described in the following lemma.

\begin{lemma}[{\cite{BKRS17}}; $\mathbf{(*)}$] \label{l:series}
For a string $s$ and $p\ge 1$, let $U$ be a $p$-series of palindromes, $r=\#U$.
There exist unique palindromes $u,v$ with $|uv| = p$, $v \ne \varepsilon$ such that one of the conditions holds:\\
1) $U=\{(uv)^{r+1}u, (uv)^ru, \ldots, (uv)^2u\}$ and $uvu$ is the head of the next series,\\
2) $U=\{(uv)^ru, (uv)^{r-1}u,\ldots, uvu\}$ and $u$ is the head of the next series,\\
3) $U =\{v^r, v^{r-1},\ldots,v\}$, $p=1$, $|v|=1$, $u=\varepsilon$.%
\end{lemma}
With the notion of series, rule \eqref{e:dp} can be rewritten as
\begin{equation} \label{e:nlogn}
\pl^j[n] = 1 + \min_U\min_{u \in U} \pl^{1-j}[n{-}|u|], \text{ where } U \text{ runs through all series of } s.    
\end{equation}
 Our  Algorithm~\ref{a:nlogn} below makes use of Lemma~\ref{l:series} to compute each internal minimum in \eqref{e:nlogn} in $O(1)$ time at the expense of some additional storage. By Lemma~\ref{l:lognseries}, this means computing of both arrays $\pl^j[1..n]$ in $O(n\log n)$ time. Algorithm~\ref{a:nlogn} is an adaptation of the algorithm of \cite{BKRS17} for palindromic length; but since it serves as a base for the linear-time solution, we provide all necessary details. To maintain series, we define an auxiliary array $\leftside[1..n]$: for $p \in [1..n]$, if $s[1..n]$ has a $p$-series, then  $\leftside[p]$ is such that $s[\leftside[p]{+}1..n]$ is the longest suffix (which is not necessarily a palindrome) of $s[1..n]$ with period $p$; otherwise, $\leftside[p]$ is undefined. 
\begin{example} \label{ex:left}
If $s[1..n]=\cdots aaabaaba$ and $p=3$, the longest 3-periodic suffix is $s[n{-}6..n]=aabaaba$ and $\leftside[3]=n-7$. If we extend $s$ by $b$, this will break period 3 and make $\leftside[3]$ undefined. If we then append $ba$, the resulting string $s\cdot bba=\cdots aaabaab\pmb{abba}$ of length $n+3$ will have a suffix-palindrome of period 3 \pmb{again}, and $\leftside[3]$ will get the new value $n-2$. 
\end{example}
\begin{remark} \label{r:obsolete}
We  do \emph{not} explicitly make $\leftside[p]$ undefined if it was defined earlier. We compute it at the iterations where a $p$-series is present. If the new value differs from the old one, we conclude that period $p$ broke since we saw the previous $p$-series. 
\end{remark}

\begin{lemma}[$\mathbf{*}$]\label{l:left}
Suppose that the iterator contains a string $s[1..n]$ having a $p$-series. Given $p$ and $|\head(p)|$, $\leftside[p]$ can be computed in $O(1)$ time. 
\end{lemma}

Internal minima in \eqref{e:nlogn} are computed in $O(1)$ time as follows. Let $U=\{(uv)^ru,\ldots,uvu\}$, where $r>1$, be a $p$-series for $s[1..n]$ (other cases from Lemma~\ref{l:series} are similar). In \eqref{e:nlogn} we compute $m=\min\{\pl^{1-j}[n{-}rp{-}|u|],\ldots,\pl^{1-j}[n{-}p{-}|u|]\}$ to update $\pl^j[n]$. Note that $s[1..n]$ ends with $(uv)^ru$ but not with $(uv)^{r+1}u$: otherwise, the latter string would belong to $U$. Then $s[1..n{-}p]$ ends with $(uv)^{r-1}u$ but not with $(uv)^ru$ and thus has the $p$-series $U'=\{(uv)^{r-1}u,\ldots,uvu\}$. Thus, at $(n{-}p)$th iteration we computed $m'=\min\{\pl^{1-j}[n{-}rp{-}|u|],\ldots,\pl^{1-j}[n{-}2p{-}|u|]\}$ to update $\pl^j[n{-}p]$ and saved $m'$ into an auxiliary array. Then $m=\min\{m',\pl^{1-j}[n-p-|u|]\}$ is computed in constant time, as required. We store all precomputed minima in two arrays $\pre^j[1..n]$, where $j\in\{0,1\}$ and each $\pre^j[p]$ is, in turn, an array $\pre^j[p][0..p{-}1]$ such that 
\begin{multline} \label{e:pre}
\pre^j[p][i] = {\min}_t\, \pl^j[t],\text{ where }  t\in[\leftside[p]..n{-}p{-}1]\ \mathbf{and}\ (t-\leftside[p])\bmod p=i \\ \mathbf{and}\ s[t{+}1..n] \text{ begins with a palindrome of minimal period } p
\end{multline}
(see the example in Fig.~\ref{f:pre}). If the minimum in \eqref{e:pre} is taken over the empty set, $\pre^j[p][i]$ is undefined. Let $\{u_1=\head(p), \ldots, u_r\}$ be a $p$-series for $s[1..n]$, $i= n-|u_1|-\leftside[p]$. Then $0\le i<p$ and $(n-|u_\ell|-\leftside[p])\bmod p=i$ for all $\ell$. By \eqref{e:pre},  $\pre^{1-j}[p][i] = \min_{\ell\in[1..r]} \pl^{1-j}[n{-}|u_\ell|]$, which is exactly the value $m$ mentioned above.  We denote $\PRE[p][i]=(\pre^0[p][i],\pre^1[p][i])$. 
\begin{figure}[htb]
    \centering
    \includegraphics[scale=0.87, trim= 35 749 125 31, clip]{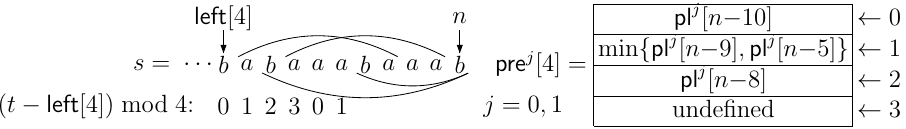}
    \caption{Precomputed values for the 4-series at the $n$th iteration; see \eqref{e:pre}. Residues modulo 4 are shown for all positions $t\in[\leftside[4]..n{-}4{-}1$]. Each position followed by a palindrome of minimum period 4 (an arc) contributes to the computation of an element of $\pre^j[4]$. The palindromes following the only position marked by 3 have minimum periods 1 ($aa$) and 3 ($aabaa$), but not 4.}
    \label{f:pre}
\vspace*{-3mm}
\end{figure}

If $s[1..n]$ has a suffix-palindrome $w$ centered at $x$, we say that $w$ \emph{survives} the next iteration if $x$ is the center of a suffix-palindrome of $s[1..n{+}1]$.  Otherwise, $w$ (or $x$) \emph{dies} at that iteration. We refer to the number of future iterations $x$ survives as its \emph{time-to-live}, denoted by $\pttl_n(x)$. The same notions apply to any series of $s[1..n]$ and to its period $p$; we write $\ttl_n(p)$ for the time-to-live of $p$. If $p$ dies, then $p$-series also dies, but not vice versa. More precisely, while $p$ is alive, $p$-series evolves\label{p:evolution} as follows (see Fig.~\ref{f:evolve}; $\mathbf{(*)}$). 
A $p$-series appears at $n$th iteration as a single palindrome $\head(p)=uvu$ centered at $x$. At the iteration $n{+}i$ such that $n+i-x=x-\leftside[p]$, $x$ dies together with the series. At some iteration $n{+}i{+}\ell$, $0\le \ell<p{-}i$, the series ``reborns'' as a palindrome centered at $x+\frac p2$  (if $\ell=0$, we say that the series has not died). At the $(n{+}p)$th iteration, a palindrome centered at $x{+}p$ joins the series; every subsequent $p$ iterations follow the same pattern, but the death of the head no longer means the death of the series. Knowing the structure of series allows one to store $\PRE[p]$ in dynamic arrays, avoiding the allocation of $\Omega(n^2)$ space. The proof of the next lemma has been moved to Appendix, since a stronger result (Lemma~\ref{l:compresspre}) is proved in Section~\ref{s:n}.
\begin{figure}[htb]
    \centering
    \includegraphics[scale=0.87, trim= 35 745 310 32, clip]{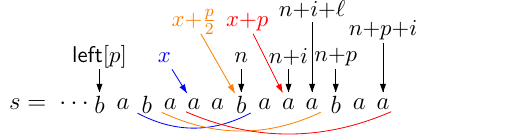}
    \caption{Evolution of a series ($p=4, i=2, j=1$). 4-series appears at the $n$th iteration as a single palindrome $\head(p)=baaab$ centered at $x$ (blue arc). At $(n{+}i)$th iteration it dies because $s[n{+}i]\ne s[\leftside[p]]$. At $(n{+}i{+}\ell)$th iteration it reappears as $\head(p)=aaabaaa$ centered at $x+\frac p2$ (orange arc). At $(n{+}p{+}i)$th iteration this center dies, but the series survives, because it was earlier joined by a palindrome centered at $x{+}p$. Currently $\head(p)=aabaaabaa$ (red arc).}
    \label{f:evolve}
\vspace*{-3mm}
\end{figure}

\begin{lemma}\label{l:storepre}
All arrays $\PRE[p]$ can be stored in a data structure requiring $O(n\log n)$ space, $O(1)$ time per deletion of an array, and amortized  $O(1)$ time per operation with any element.
\end{lemma}

Now we are ready for Algorithm~\ref{a:nlogn}. Given a new letter $s[n]$, we compute $\PL[n]$ as follows:
\begin{algorithm*}
\caption{: $O(n\log n)$ algorithm, $n$th iteration}
\label{a:nlogn}
\begin{algorithmic}[1]
\small
\State $\append(s[n]);\;\PL[n] \gets (\infty,\infty);$
\For{($x \gets \maxPal;\; x \ne n + \frac{1}2;\; x \gets \nextPal(\cntr(d))$)} \Comment{goes to next head each time}
	\State $p \gets \len(x) - \len(\nextPal(x));$\Comment{min. period of the head centered at $x$}
    \State $d \gets p + (\len(x) \bmod p);$\Comment{length of candidate $\tail(p)$}
    \If{$\len(\cntr(d)) - \len(\nextPal(\cntr(d))) \ne p$}
    	  $d \gets d + p;$\Comment{corrected length of $\tail(p)$}
    \EndIf
    \State $y\gets\leftside[p]$; compute $\leftside[p]$; $i\gets n{-}\len(x){-}\leftside[p]$;\Comment{$O(1)$ time by Lemma~\ref{l:left}}
    \If{$\leftside[p]>y$}
    	 clear $\PRE[p]$;\Comment{delete obsolete values}
    \EndIf
    \If{$\len(x) = d$}\label{lst:onel:seriesCond}
    	$\PRE[p][i] \gets \PL[n{-}d];$
    \EndIf
    \State $\PRE[p][i] \gets (\min\{\pre^0[p][i], \pl^0[n{-}d]\}, \min\{\pre^1[p][i], \pl^1[n{-}d]\})$;\label{lst:calcPrecalc}
    \State $\PL[n] \gets (\min\{\pl^0[n], 1 + \pre^1[p][i]\}, \min\{\pl^1[n], 1 + \pre^0[p][i]\});$
\EndFor
\end{algorithmic}
\end{algorithm*}

\begin{proposition} \label{p:nlogn}
Algorithm~\ref{a:nlogn} correctly computes $\PL[1..n]$ in $O(n\log n)$ time.
\end{proposition}
\begin{proof}
All  $\append$ queries require $O(n)$ time in total, so we ignore them. Let $x$ be the center of a processed suffix-palindrome $w=(uv)^ru$ with period $p$. Then $p =|uv|=\len(x) - \len(\nextPal(x))$ (see, e.g., \cite[Lemmas 2,3]{KosolobovPalk}). Let $x' = \cntr(p + (\len(x) \bmod p))$ be the center of the suffix-palindrome $uvu$, $p'=\len(x') - \len(\nextPal(x'))$. By Lemma~\ref{l:series}, $uvu=\tail(p)$ if $p'=p$ and $uvu=\head(p')$ otherwise. Thus in lines 3--5 Algorithm~\ref{a:nlogn} computes, using $O(1)$ queries to the iterator, the period and the length of $\tail(p)$. This means, in particular, that the \textbf{for} loop iterates over all heads of series in $s$, which means $O(\log n)$ runs by Lemma~\ref{l:lognseries}.

If a symbol added to $s$ breaks period $p$, all values in $\pre^j[p]$ become obsolete and should be deleted. Algorithm~\ref{a:nlogn} handles this in lines 6--7, using Remark~\ref{r:obsolete}. 

Let $\{u_1, \ldots, u_r\}$ be a $p$-series, $i=n{-}|u_1|{-}\leftside[p]$. If $r=1$, there was no $p$-series $p$ iterations ago, so the undefined value  $\pre^j[p][i]$ is set to $\pl^j[n{-}|u_r|]$ in line~\ref{lst:onel:seriesCond}. Otherwise one has $\pre^j[p][i] = \min\{\pl^j[n{-}|u_1|], \ldots, \pl^j[n{-}|u_{r-1}|]\}$ by \eqref{e:pre}, and this value is updated using $\pl^j[n{-}|u_r|]$ in line~\ref{lst:calcPrecalc}; so $\pre^j$ is correctly maintained. Finally, in line 10 the rule \eqref{e:nlogn} is implemented. So the algorithm is correct and each run of the \textbf{for} loop takes $O(1)$ amortized time due to Lemma~\ref{l:storepre}. The result now follows.
\end{proof}

\section{Linear-Time Algorithm} \label{s:n}

\subparagraph*{Resources for speed-up.} 
In some cases, dynamic programming can be sped up by a $\log n$ factor by a technique called \emph{four Russians' trick} \cite{ADKF70}. The idea is to store the DP array(s) in a compressed form requiring $O(1)$ bits per element and update $(\log n)$-size chunks of the compressed array using $O(1)$ operations on machine words. The key operations used are \emph{table operations}: $f(..)$ is a table operation if it has $o(n)$ (typically $O(n^\alpha)$, where $\alpha<1$) valid inputs and the results for all valid inputs can be computed in $o(n)$ time. 
\begin{remark}
We follow a usual scheme: all results for a constant number of table operations are computed in the $o(n)$-time preprocessing phase and stored in auxiliary tables.
\end{remark}
This technique was used, in particular, for palindromic factorization in \cite{KosolobovPalk, BKRS17} and for square factorization in \cite{MIBTM16}. To apply it to Algorithm~\ref{a:nlogn}, we should meet the following conditions: 
(i) the array $\PL$ can be compressed to $O(1)$ bits per element; (ii) each array $\PRE[p]$ can be compressed to $O(p+\log n)$ bits; (iii) updates of arrays $\PL,\PRE[p]$ (lines 8--10) can be performed without decompression, simultaneously for $\Omega(\log n)$ successive iterations with a constant number of operations over machine words; (iv) all intermediate states of arrays $\PL$ and $\PRE[p]$ during the course of the algorithm are valid inputs for the compression scheme.

For palindromic length \cite{BKRS17} this works as follows. 
If $w$ is a string and $a$ is a letter, then $|\pl(wa)-\pl(w)|\le 1$ \cite[Lemma 4.11]{RuSh18}. Thus it is possible to store any subarray $\pl[i..j]$ as one number $\pl[i]$ followed by $(j{-}i)$ 2-bit codes for the differences $\pl[r]-\pl[r{-}1]$. The situation with the arrays $\pre[p]$ is subtle, but each of these arrays can be efficiently split into a constant number of chunks, where successive elements of the same chunk differ by at most one. For each chunk, the same encoding as for $\pl$ works. 
Range updates are based on the observation that if $s[1..n{-}1]$ has a suffix-palindrome with the center $x$, which survives next $t$ iterations, then one can assign $\pl[n..n{+}t{-}1] \gets \min\{\pl[n..n{+}t{-}1], 1 + \overleftarrow{\pl[2x{-}n{-}t{-}1..2x{-}n]}\}$.

Let $t=\lfloor \frac{\log n}{8}\rfloor$. Formally, a \emph{chunk} is an array $A=A[1..h]$, where $h<t$, of  $(\log n)$-bit numbers such that  $|A[i]-A[i{-}1]|\le 1$ for all $i$; it is  stored as a $(\log n)$-bit number followed by $h$ $2$-bit codes encoding these differences. If the length of a chunk is less than $t$, the unused 2-bit code is added to the end.
In chunks $\lvec{A}$ and $\min\{A,B\}$ the consecutive elements differ by at most 1, so they can be compressed. The next lemma allows fast performance.

\begin{lemma}[\cite{BKRS17};$\mathbf{(*)}$] \label{l:tableop}
The following operations can be performed in $O(1)$ time using table operations: (1) incrementing all elements of a chunk, (2) extracting an element from a chunk, (3) extracting a chunk at any position, (4) reversing a chunk, (5) concatenating two chunks, (6) extending a chunk with dummy values, (7) taking the minimum of two chunks. 
\end{lemma}

The algorithm of \cite{BKRS17} groups iterations into \emph{phases}. Each phase begins immediately after the end of the previous phase and continues until one of three conditions is met: $t$ iterations passed, the input string ended, or the longest suffix-palindrome will die at the next iteration.
In the beginning of a phase, a ``prediction'' is made that $\maxPal$ survives the next $t$ iterations. Under this assumption, time-to-live's of periods are computed in $O(1)$ time and range updates to $\pre$ and $\pl$ are performed for the corresponding number of iterations, using $O(1)$ operations from Lemma~\ref{l:tableop}. After processing all series, actual $t$ letters are added one by one; each time the iterator is updated, one or two new centers for palindromes appear. These palindromes are used to update $\pl$ and $\pre$; their time-to-live's are computable in $O(1)$ time. 
When $s[i]$ is processed, $\pl[i]$ gets its true value. If an input symbol changes $\maxPal$, the phase is aborted, unfinished updates are deleted, and a new phase is started from the current symbol. 

\subparagraph*{Compression and smoothing.} 

The following lemma allows one to compress the array $\PL$.

\begin{lemma} \label{l:compress}
If $w$ is a string, $a,b$ are letters, $j\in\{0,1\}$, then $\pl^j(wab)\in\{\pl^{1-j}(wa)+1$, $\pl^{1-j}(wa)-1, \pl^{1-j}(w)+1\}$.
\end{lemma}
\begin{proof}
Let $k=\pl^j(wab)$. Consider all palindromic factorizations of the form $wab=w_1\cdots w_k$. Three cases are possible.\\
1. There is a factorization with $w_k=b$. Then $k = \pl^{1-j}(wa)+1$.\\
2. There is a factorization with $w_k=bub$, $u\ne\varepsilon$, and no factorization with $w_k=b$. Then $wa$ has no $(k{-}1)$-factorization, but has a $(k{+}1)$-factorization $w_1\cdots w_{k-1}bu$; so $k = \pl^{1-j}(wa)-1$.\\ 
3. $w_k=ab$ in each factorization (so $a=b$). Then $k = \pl^{1-j}(w)+1$.
\end{proof}

As above, we set $t=\lfloor \frac{\log n}{8}\rfloor$. \emph{Double chunks} (called just \emph{chunks} if no confusion arises) are segments $A=A[1..h]=(A^0[1..h],A^1[1..h])$ of $\PL$ or $\PRE[p]$. We encode them using $O(t)=O(\log n)$ bits. The difference with the case of palindromic length is that we store four explicit values: $A^0[1],A^1[1],A^0[2],A^1[2]$. Subsequent elements are encoded by 2-bit codes associated with the cases of Lemma~\ref{l:compress}, the unused 2-bit code indicates the end of a shorter chunk. Below we demonstrate a problem with this encoding and show the solution. 

Let us consider the graphs of two functions $f_1(i)=\pl^{i\bmod 2}[i]$ and  $f_2(i)=\pl^{(i+1)\bmod 2}[i]$ ($f_2=\infty$ for small values of $i$ until two consecutive equal letters occur in $s$). At most points, $|f_j(i)-f_j(i{-}1)|=1$; if $f_j(i)-f_j(i{-}1)<-1$ (and thus $f_j(i)=f_{3-j}(i{-}2)$ by Lemma~\ref{l:compress}, case 3), we say that $f_j$ has a \emph{drop at $i$}. In the left graph in Fig.~\ref{f:zigzags}, drops are dash lines.

\begin{figure}[htb]
\centering
\vspace{-2mm}
\includegraphics[scale=0.9, trim= 30 695 300 31, clip]{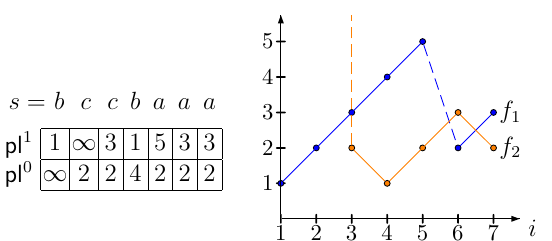}
\includegraphics[scale=0.9, trim= 140 695 300 31, clip]{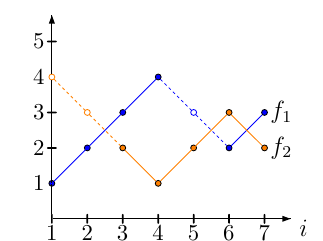}
\caption{\small even/odd palindromic length visualized by $f_1(i)=\pl^{i\bmod 2}[i]$ and  $f_2(i)=\pl^{(i+1)\bmod 2}[i]$. On the right, the graphs of $f_1$ and $f_2$ after smoothing the $\PL$ array.}
\label{f:zigzags}
\vspace{-2mm}
\end{figure}

The drops cause the following problem: taking minimum of two chunks can result in a chunk having no valid encoding, violating condition (iv). For example, taking the minimum of two valid chunks
\arraycolsep=2pt
$\small
\begin{array}{|c|c|c|c|}
\hline
\cdots&3&7&3\\
\hline
\cdots&6&4&4\\
\hline
\end{array}
$ and
$\small
\begin{array}{|c|c|c|c|}
\hline
\cdots&1&7&3\\
\hline
\cdots&6&2&6\\
\hline
\end{array}
$\,, we get the chunk 
$\small \begin{array}{|c|c|c|c|}
\hline
\cdots&1&7&3\\
\hline
\cdots&6&2&\pmb{4}\\
\hline
\end{array}$\,. The marked element does not satisfy any alternative from Lemma~\ref{l:compress}. To remedy this, we consider the following \emph{smoothing} operation on the computed array $\PL[1..n{-}1]$: for each $i\in[1..n{-}2]$, $j\in\{0,1\}$, replace $\pl^j[i]$ by $\min\{\pl^j[i], \pl^{(j+1)\bmod 2}[i{+}1]+1,\ldots,\pl^{(j+n-1-i)\bmod 2}[n{-}1]+n-1-i\}$. The result of smoothing can be seen in Fig.~\ref{f:zigzags} (right graph). The main property of smoothing is  
\begin{lemma} \label{l:smooth}
Smoothing of the array $\PL[1..n{-}1]$ does not affect the value $\PL[n]$ computed by rule \eqref{e:dp}.
\end{lemma}
\begin{proof}
Assume that we applied smoothing and then computed $k=\pl^{k\bmod 2}[n]$ by rule \eqref{e:dp} getting the minimum $k{-}1$ as the element $\pl^{(k-1)\bmod 2}[r]$. Since the minimum cannot increase after smoothing, we just need to check that $s[1..n]$ indeed has a palindromic $k$-factorization. This is obvious if $\pl^{(k-1)\bmod 2}[r]$ was not changed by smoothing. Otherwise, $k{-}1=\pl^{(k-1+i-r)\bmod 2}[i]+i-r$ for some $i>r$. Since $s[r{+}1..n]$ is a palindrome (see \eqref{e:dp}), we factor it as $uv\lvec{u}$, where $u=s[r{+}1..i]$. So, $s[1..n]=s[1..i]v\lvec{u}$ indeed has a $k$-factorization: $s[1..i]$ can be factored into $k{-}1{-}i{+}r$ palindromes and $\lvec u$ into $i{-}r$ 1-symbol palindromes.
\end{proof}

\begin{remark} \label{r:smooth}
1) We extend smoothing to subarrays of the form $\PL[l..r]$: for each $i\in[l..r{-}1]$, $j\in\{0,1\}$, replace $\pl^j[i]$ by $\min\{\pl^j[i], \pl^{(j+1)\bmod 2}[i{+}1]+1,\ldots,\pl^{(j+r-i)\bmod 2}[r]+r-i\}$. The resulting values are always between the values from the original $\PL$ and those from smoothed $\PL$, so   Lemma~\ref{l:smooth} stays true for this \emph{local} smoothing.\\
2) From formula \eqref{e:pre} one can conclude that smoothing also works for the arrays $\PRE[p]$. Namely, if $\pre^j[p][i]$ and $\pre^{1-j}[p][i{-}1]$ are used for computing $\pl^{1-j}[n]$ and $\pl^j[n{+}1]$ respectively (i.e, a $p$-series survives during these two iterations), then  $\pre^{1-j}[p][i{-}1]$ can be replaced by $\min\{\pre^{1-j}[p][i{-}1], \pre^j[p][i]+1 \}$. This replacement can be iterated for a range.\\
3) Observations 1), 2) allow one to use the following principles: extracting a chunk from $\PL$ or $\PRE$, smooth it; concatenating two chunks in $\PRE$, smooth the result; taking the minimum of two smoothed chunks of different length, extend the shorter chunk with dummy values satisfying $f_j(i)-f_j(i{+}1)=1$ (resp., $=-1$) for left (resp., right) extensions.
\end{remark}
 
The success of the approach described in Remark~\ref{r:smooth} relies on two lemmas. 

\begin{lemma} \label{l:O1smooth}
Smoothing a (double) chunk can be done in $O(1)$ time.
\end{lemma}
\begin{proof}
Suppose that a chunk $A[1..h], h\le t$, is given, so we know the numbers $f_1(1), f_2(1)$, $f_1(2),f_2(2)$, and the 2-bit codes $b_3^1,\ldots,b_h^1,b_3^2,\ldots,b_h^2$, where the functions $f_1,f_2$ are defined as above and $b_i^j$ encodes the case from Lemma~\ref{l:compress} for $f_j(i)$. We write $\tilde{A},\tilde{f}_1,\tilde{f}_2$ for the chunk and the functions after smoothing. If $f_j$ has a drop at 2, we can replace $f_j(1)$ with $f_j(2){+}1$ without affecting $\tilde{f}_j$. With this reservation, we can restore $f_1,f_2$ from $f_1(2),f_2(2), b_2^1,\ldots, b_h^1, b_2^2,\ldots,b_h^2$. Indeed, compute $f_j(1)$ for $j=1,2$ from $f_j(2)$ and $b_2^j$; then consider $f_j(3)$. If $f_j$ has no drop at 3, it is computed from $f_j(2)$, $b_3^j$. If it has this drop, then $f_{3-j}(2)\le f_{3-j}(1)+1=f_j(3)\le f_j(2)-3\le f_j(0)-1$ (by 0 we mean the position preceding the first position in $A$), and so $f_{3-j}$ has no drop at 2. Then the computed value $f_{3-j}(1)$ is true, and we put $f_j(3)=f_{3-j}(1)+1$. Knowing $f_j(2)$ and $f_j(3)$, it is easy to reconstruct the rest.

For any constant $c$, one can obtain $\tilde{f}_j+c$ from $f_j+c$. So to get a table operation we replace $f_2(2)$ with 0 and $f_1(2)$ with $\delta= f_1(2)-f_2(2)$. Now observe that if $|\delta|\ge 2h$, only two cases are possible, and they are easy to distinguish: either there are no drops, so $\tilde{f}_j=f_j$ for $j=1,2$, or the leftmost drop is in $f_j$ at position $i$, and then  $\tilde{f}_j(i-r)=f_j(i)+r$ for all $r<i$. So in this case only the sign of $\delta$ matters, and we assign $\delta=\pm\infty$. Thus we have got a table operation which, given $(2t-2)$ 2-bit codes and a number $\delta$ with $O(t)$ distinct values, returns a compressed chunk; adding $f_2(2)$ to the explicit values in this chunk, one gets $\tilde{A}$.
\end{proof}

\begin{lemma}[$\mathbf{*}$] \label{l:tableop2}
The following operations with (double) chunks can be performed in $O(1)$ time: (1) incrementing all elements of a chunk, (2) extracting an element from a chunk, (3) extracting a chunk at any position, (4) concatenating two chunks, (5) extending a smoothed chunk with dummy values, (6) taking the minimum of two smoothed chunks. 
\end{lemma}

\subparagraph*{Algorithm~EOPL.}
We describe one phase of Algorithm~EOPL computing $\pl^0(s)$ and $\pl^1(s)$ for a string $s$. 
The definition of phase and the idea to ``predict'' the next $t$ input symbols and perform range updates by means of table operations are the same as for palindromic length.

\smallskip\noindent
\textit{Prerequisites.} 
P1. Arrays $\PL$ and $\PRE[p]$ are stored as sequences of compressed length-$t$ chunks (the last chunk in a sequence can be short). We maintain a \emph{work chunk} $W$ to compute the reversal of a new chunk of $\PL$ during the current phase; we move symbols from $W$ to $\PL$ one by one and thus avoid the reversal operation. All chunks applied to $W$ are aligned to its current right end (corresponding to the current iteration). List $\wait$ stores new palindromes and palindromes that changed periods, to perform $\PRE$ updates at the end of the phase.\\
P2. Using Lemma~\ref{l:smooth} and Remark~\ref{r:smooth}, we smooth chunks extracted from $\PL$ and extend short smoothed chunks before applying $\min$ operations. Each array $\PRE[p]$ consists of smoothed chunks and $W$ is also smoothed; drops may occur between the last element of $\PL$ and the first element of $W$ as a result of processing 2-letter suffix-palindromes at step F3a.\\
P3. \emph{Big} $p$-series ($p\ge t$) and \emph{small} $p$-series ($p<t$) are processed separately.

\smallskip\noindent
\textit{First iteration.}
F1. Read the next symbol (say, $s[n]$), perform $\append(s[n])$; at this moment\\
a. the iterator stores $s[1..n]$ and returns the list $\dead$;\\
b. the array $\PL[1..n{-}1]$ and the arrays $\PRE[p]$ for $s[1..n{-}1]$ are correctly computed;\\
c. the chunk $W$ is initialized by $\infty$'s, the list $\wait$ is empty;\\
d. the maximum number of iterations in the phase is set to $t'=\min\{n-\len(\maxPal), t\}$ (if the first number is smaller, $\maxPal$ must change after at most $t'$ iterations).\\ 
F2. Loop through the list of big $p$-series of $s[1..n{-}1]$:\\
a. compute $\leftside[p]$ and clear $\PRE[p]$ if necessary;\\
b. compute $\ttl'[p] = \min\{\ttl_{n-1}[p],t'\}$ using symmetry inside palindromes;\\
c*. using $\leftside[p]$ and $\ttl'[p]$, update $\PRE[p]$, $W$, and $\wait$ (details in a separate item below).\\
F3. Process new suffix palindromes with centers $n-\frac 12$ (if $s[n{-}1]=s[n]$) and $n$:\\
a. update $W$ directly using the chunks of $\PL$ ending at position $n{-}1$ (and possibly $n{-}2$);\\
b. using symmetry again, check whether $n$ and/or $n-\frac 12$ remain centers of suffix-palindromes after the last iteration of the phase; add the center(s) with the answer ``yes'' to the list $\wait$.
\\
F4. Small series and finalization:\\
a. extract from $W$ its last element $m= (m_0,m_1)$ and delete it from $W$;\\
b. assign $(\pl^0[n],\pl^1[n])=(m_0,m_1)$ and process small $p$-series of $s[1..n{-}1]$ directly by Algorithm~\ref{a:nlogn}, getting the final value of 
$\PL[n]$;\\
c. update arrays $\PRE[p]$ with all palindromes from the list $\dead$.

\smallskip\noindent
\textit{Subsequent iterations.}
S1. Read $s[i]$, compare it to $s[i{-}1{-}\len(\maxPal)]$.\\
S2. If $s[i]$ extends the longest suffix-palindrome, the phase continues: perform $\append(s[i])$ and then the operations from steps F3 and F4, replacing $n$ by $i$.\\
S3. If $s[i]$ breaks the longest suffix-palindrome, the phase is aborted:\\
a. update the lists $\PRE[p]$ using suffix-palindromes from the list $\wait$, clear $\wait$;\\
b. start the next phase with $\append(s[i])$.\\
S4. If there was no abortion and all $t'$ symbols are processed, the phase is terminated:\\
a. update the lists $\PRE[p]$ using suffix-palindromes from $\wait$, clear $\wait$;\\
b. start the next phase with $\append(s[n{+}t'])$. 

\smallskip\noindent
\textit{F2c detailed.} 
Let $x$ be the center of $\tail(p)$, $\pttl'(x)=\min\{\pttl_{n-1}(x),t'\}$, 
$d= \min\{\ttl'(p), \pttl'(x)\}$. 
Update $\PRE[p]$ with a length-$d$ chunk from $\PL$ and $W$ with the updated length-$d$ chunk from $\PRE[p]$. If $\ttl'(p) < \pttl'(x)$ (after $d$ iterations period $p$ dies, but $\head(p)=\tail(p)$ survives, becoming a palindrome with a bigger period), additionally update $W$ with a chunk of length $\pttl'(x)$; if $\pttl'(x)\ge t'$, add $x$ to $\wait$. Note that if $p$-series gets new tail $x{+}\frac p2$ during the phase, then at the start of the phase $x{+}\frac p2$ was the center of the head of some $p'$-series, survived the death of that series, and thus was used for an additional update of $W$. Necessary updates of $\PRE[p]$ are carried when the lists $\dead$ and $\wait$ are processed. All the same stays true for $x{+}p$ which is the next potential tail of the $p$-series.

\begin{lemma}\label{l:lincorrect}
Algorithm EOPL correctly computes the array $\PL$.
\end{lemma}

\begin{proof}
The computation of $\PL$ by Algorithm~EOPL differs from the correct computation by Algorithm~\ref{a:nlogn} in a few points. The use of smoothed chunks is justified by Lemmas~\ref{l:smooth},~\ref{l:O1smooth}, so below we take smoothing into account speaking about correctness of the arrays $\PRE[p]$. We prove the lemma by induction, with an obvious base and F1b as the hypothesis; more precisely, we assume that at the start of a new phase with $s[n]$, the array $\PL[1..n{-}1]$ and all arrays $\PRE[p]$ are correct, where $p$ runs through the set of live periods of $s[1..n{-}1]$ (i.e., $s[1..n{-}1]$ has a $p$-periodic suffix with at least one $p$-periodic palindrome in it). 

In an aborted phase, some predictions in F2c are made beyond the actual end of the phase (the phase processes $s[n..n{+}i{-}1]$, and updates are made for more than $i$ elements of $W$ and $\PRE[p]$). For $W$, this does not matter, because $W$ is translated to $\PL$ one element per iteration (F4a) and is initialized at the beginning of each phase (F1c). For $\PRE[p]$, the situation means that $p$ will become dead at the $(n{+}i)$th iteration: the actual symbol $s[n{+}i]$ differs from the predicted symbol which preserved the period. 

Next, if $p$ is alive for $s[1..n{+}i{-}1]$, all necessary updates for $\PRE[p]$ were made at steps F2c, F4c (useful if a palindrome dies while its period survives), and S3a/S4a. Note that repeated updates using the same palindrome (say, first at step F2c and then at F4b) cannot harm.
Finally, all palindromes ending in $s[n{-}1..n{+}i{-}2]$ were used to update $W$, so $\PL[n..n{+}i{-}1]$ is computed according to rule \eqref{e:dp}.  The result now follows.
\end{proof}



\subparagraph*{Details and Analysis of Algorithm~EOPL.} To prove Theorem~\ref{MainTheorem}, it remains to show that the details of computation can be organized to provide the linear-time performance. The key part is maintaining $\PRE[p]$ arrays.

\begin{lemma} \label{l:compresspre}
All arrays $\PRE[p]$ can be stored in a compressed form in a data structure requiring $O(n)$ space, $O(1)$ time per deletion of an array, and amortized $O(1)$ time per operation with any chunk.
\end{lemma}
\begin{proof}
We use dynamic arrays (like \emph{vectors} in C$^{++}$). Such an array has \emph{size} (space in use) and \emph{capacity} (allocated space). When the size increases and reaches the capacity, the latter doubles; this step rebuilds the array to provide a constant-time access to its elements and takes the time linear in its size. Thus, adding an element can be done in amortized $O(1)$ time, an existing element can be modified in $O(1)$ time, an array can be cleared by setting its size to 0, and the total allocated space is proportional to the maximum size reached.

For each array $\PRE[p]$, we store integers $R_1, I_1, R_2, I_2$ and dynamic arrays $F_1$ and $F_2$, both initially of size 0 and capacity 1. We also perform (amortized) constant-time chunk operations listed in Lemma~\ref{l:tableop2}. When we need to extract a chunk for updating the array $W$, we take the minimum of the corresponding chunks from $F_1$ and $F_2$. We refer to the evolution of $p$-series described at p.~\pageref{p:evolution}. Note that $\PRE[p]$ is always filled right to left (see  Fig.~\ref{f:pre},\,\ref{f:evolve}): starting with the index $i{-}1$ in the first pass, with the index $p{-}\ell{-}1$ in the second pass, and with the index $p{-}1$ in subsequent passes.
During its time-to-live, $p$-series is processed in three stages. To decide the current stage, the indicators described below are used. \\
1. A $p$-series appears at $n$th iteration as a palindrome $uvu$ centered at $x$; at this or one of subsequent iterations we get a first chunk to add to $\PRE[p]$. We set $R_1$ to the index of the first element added to $\PRE[p]$ at $n$th iteration; it is $\ttl_n(x)$ or $i-1$ in terms of Fig.~\ref{f:evolve}. Then we write the first chunk to $F_1$ and set $I_1$ to the last used index. After that, we possibly add more chunks until $x$ dies, updating $I_1$ respectively. (Each chunk is concatenated with the previous one such that all of them except the last one have length $t$.) At this moment, $I_1=0$ and the series dies. The indicator of stage 1 is $(\mathit{size}(F_2)=0) \wedge (I_1>0)$.\\
2. At one of subsequent iterations ($n{+}i{+}\ell$ in Fig.~\ref{f:evolve}) the series reborns with the new head/tail $x{+}\frac p2$. Getting the next chunk, we set $R_2$ to the index $p{-}\ell{-}1$, write the chunk, set $I_2$ to the last used index, and proceed with subsequent chunks, using concatenation and updating $I_2$. At some point we get $I_2\le R_1{+}1=i$. This means that the tail of the series has changed to the palindrome centered at $x{+}p$ either during the phase or immediately after it. All updates made after this change make no sense (we have written to $F_2$ the elements already written to $F_1$); so we set $I_2=i$ and wait for the next chunk (if the series survives the current phase, the chunk will appear when the list $\wait$ will be processed). From this point, we take minimum of the new chunks with the corresponding chunks of $F_1$ and then add the result to $F_2$. The stage continues until $I_2=0$. The indicator of stage 2 is $(I_1=0)\wedge (R_1<n-1)$.\\
3. Getting $I_2=0$ at the previous stage, we set $R_1=n-1, I_1=n$, and clear $F_1$. All subsequent chunks are written to $F_1$, in the way described in the previous stages. When we get $I_1\le i$ for the first time (change of the tail of the series), we set $I_1$ to $i$ and truncate the last chunk. When $I_1=0$ is reached, we reset $I_1=p$ and set a flag telling that each next chunk should be written as the minimum of the new chunk and the existing chunk at the same positions of the array. The indicator of stage 3 is $R_1=n-1$.

Thus, we perform each update in $O(1)$ time (amortized, due to the properties of dynamic arrays). For $j=1,2$, all chunks in $F_j$, except for the last one, have length $t$, so it is easy to find the argument for the ``extract a chunk'' operation in $O(1)$ time using $R_j$. We finally note that the total number of stored chunks is $O(n)$, each requiring $O(1)$ machine words.
\end{proof}


Let us prove time bounds for all steps of Algorithm~EOPL. All $\append$ queries require $O(n)$ time in total. By Lemma~\ref{l:left}, F2a requires $O(1)$ time; F2b is covered by Lemma~\ref{l:O1ttl} below.

\begin{lemma} \label{l:O1ttl}
For a big $p$-series, the value $\ttl'[p]$ can be computed in $O(1)$ time.
\end{lemma}
\begin{proof}
Let $x$ be the center of the head $(uv)^ru$ of the $p$-series. We compute $|u|=\len(x)\bmod p$ and $y=\cntr(|u|)$ as in Lemma~\ref{l:left}; $y'=2\maxPal-y$ is the center of the prefix-palindrome $u$ of the longest suffix-palindrome $s[i..n{-}1]$ of $s[1..n{-}1]$. Note that the longest suffix-palindrome of $s[1..n{+}t'{-}1]$ is $z=s[i{-}t'..n{+}t'{-}1]$ because $\maxPal$ does not change during a phase. Let $w$ be the longest palindrome with the center $y'$; using the query $\rad(y')$ we determine whether $w$ is located inside $z$ or begins outside it. In the latter case, $z$ has a $p$-periodic prefix ending with $(uv)^ru$ and then a $p$-periodic suffix beginning with $(uv)^ru$. So $\ttl_{n-1}[p_i]\ge t'$ and then $\ttl'[p_i]=t'$. In the former case, $w=\lvec{w}_1uw_1$, where $w_1$ is a proper prefix of $vu$ because $|vu|=p_i>t'\ge |w_1|$. Hence the period $p_i$ breaks on the left of $w$, which means that $\ttl_{n-1}[p_i]=|w_1|=(\len(y')-\len(y))/2$. All computations above take $O(1)$ time. The result now follows.
\end{proof}
F3b uses the same idea as Lemma~\ref{l:O1ttl}: for the center $n{+}i$, compare $\rad(2\maxPal-n-i)$ to the number of remaining iterations. Both F3a and F4a require, per iteration, $O(1)$ constant-time chunk operations from Lemma~\ref{l:tableop2}. Each update in step F4b can be also performed by $O(1)$ such operations: extract elements of $\PRE[p]$ and $\PL$, take minima, and ``return'' the updated element to $\PRE[p]$, extending 1-element chunk and using $\min$. Hence, by Lemma~\ref{c:short} below, F4b takes $O(p+i)$ time per phase of $i$ iterations, where $p$ is the longest small period.

\begin{lemma}\label{c:short}
In a phase of $i$ iterations, the number of times the \textbf{for} loop of Algorithm~\ref{a:nlogn} processes series which existed at the start of the phase and have periods $\le p$, is $O(p+i)$.
\end{lemma}
\begin{proof} We prove two claims.\\[1mm]
\textbf{Claim 1.} If $s[1..n]$ has $p$-series and $q$-series such that $p>q$ and $\ttl_n(p)\ge p$, then $\ttl_n(q)<p$.

Let $\head(p)=(uv)^ru$. If $\ttl(q)\ge p$, the string $s[n{-}q{+}1..n{+}p]$ of length $p+q$ has periods $p$ and $q$. Hence it has period $d=\gcd(p,q)$ by the Fine--Wilf theorem \cite{FineWilf}. Thus $vu$ is a $(p/d)$-power of a shorter word; so $(uv)^ru$ has period $d<p$, contradicting the definition of $p$-series. The claim is proved.\\[1mm]
\textbf{Claim 2.} 
Let a string $s[1..n]$ have series with periods $p=p_1>p_2>\cdots >p_\ell$ and let $i>0$. Then $\sum_{c=1}^{\ell} \min\{\ttl_n(p_c),i\} =O(p+i)$.

Divide the periods in two groups: those with $\ttl_n(p)<p$ and the rest. In the first group, the sum of $\ttl$'s is upper bounded by $\sum_{c=1}^\ell p_c$, which is $O(p)$ by Lemma~\ref{l:lognseries}. For the second group, $\ttl_n(p_c)$ is smaller than the previous period from this group by Claim 1. So we can take $i+\sum_{c=1}^\ell p_c=O(p+i)$ as the upper bound. The claim now follows.

The statement of the lemma is immediate from Claim 2 and definitions. 
\end{proof}

Let $P$ be the minimum period of the longest suffix-palindrome of $s$ at the beginning of the phase. In F2c, the updates to $W$ and $\wait$ require $O(1)$ time per series, which is $O(\log P)$ per phase by Lemma~\ref{l:lognseries}. The updates to $\PRE[p]$ arrays are considered in Lemma~\ref{l:O1rangeupdates} below.

\begin{lemma} \label{l:O1rangeupdates}
During a phase of length $i$, all range updates of $\PRE[p]$ at steps F2c and S3a/S4a spend $O(\log P+i)$ time in total. 
\end{lemma}
\begin{proof}
Altogether, there are $O(\log P)$ series of palindromes in a phase, $O(\log P)$ (from F2c) plus $O(i)$ (from F3b) palindromes in the list $\wait$ in a phase. 
Due to Lemma~\ref{l:tableop2}, it suffices to prove that each update requires $O(1)$ chunks. For big periods, this is obvious from the description of Algorithm~EOPL. Consider small periods. For each palindrome $w$, it is enough to update the range of $\PRE[p]$ corresponding to the iterations where $w=\tail(p)$. If this range begins with $\PRE[p][\ell]$ then it ends not later than $\PRE[p][0]$ is reached (recall that $\PRE[p]$ is filled right to left). Thus we can cut the range computed from time-to-live so that it will fit into one chunk. The lemma now follows.
\end{proof}  

Finally, for the updates of $\PRE$ at step F4c, the argument of Lemma~\ref{l:O1rangeupdates} about $O(1)$ chunks per update also works. All lists $\dead$ for the total run of Algorithm~EOPL contains $O(n)$ palindromes, because the iterator works in $O(n)$ time. Thus, for F4c we get the cumulative $O(n)$ time bound. For all remaining steps, the time bound we proved is $O(\log P+i)$ per phase, except $O(p+i)$ per phase in Lemma~\ref{c:short}. If $i=t$ (the phase is terminated), then we spend $O(t)$ time for $t$ iterations. Completing a short phase (either aborted or satisfying $t'<t$), we increase $\maxPal$ by at least $P/2$ in time $O(P)$. Since $\maxPal$ never decreases, we get the overall $O(n)$ time bound. Theorem~\ref{MainTheorem} is proved.


\section{Computing $k$-factorization}

A run of linear-time Algorithm~EOPL leaves the iterator containing the input string $s=s[1..n]$ and also the arrays $\pl^0[1..n]$, $\pl^1[1..n]$ for $s$. Let us use them to construct an explicit $k$-factorization in additional linear time. Let $k'=\pl^{k\bmod 2}[n]$. By Lemma~\ref{l:kvslen}, it suffices to build a $k'$-factorization of $s$ in $O(n)$ time, which can be done as follows.

\begin{algorithm*}
\caption{: Computing $k'$-factorization}
\label{a:fact}
\begin{algorithmic}[1]
\small
\For{($i \gets 1; i<k'; i{+}{+})$} build a list $\lst[i]$ of all positions   $j: \pl^{i\bmod 2}[j]=i$
\EndFor
\State $\mathit{end}\gets n$; $\mathit{factors}\gets \{\}$
\For{($i\gets k'-1; i>1; i{-}{-}$)} \Comment{main cycle: building the $k'$-factorization}
    \For{$j\in \lst[i]$}
        \If{$\len(\frac{j+1+\mathit{end}}{2}) \ge \mathit{end}-j$} break \Comment{$s[j{+}1..\mathit{end}]$ is a palindrome}
        \EndIf
    \EndFor
\State add $s[j{+}1..\mathit{end}]$ to $\mathit{factors}$; $\mathit{end}\gets j$
\EndFor
\State add $s[1..\mathit{end}]$ to $\mathit{factors}$ \Comment{after last assignment in line 6, $\mathit{end}\in \lst[1]$}
\end{algorithmic}
\end{algorithm*}

Algorithm \ref{a:fact} discovers factors of a $k'$-factorization of $s$ in reversed order, using an observation that if a string $w$ has even (odd) palindromic length $i$ then $w=uv$, where $u$ has odd (resp., even) palindromic length $i{-}1$ and $v$ is a nonempty palindrome. All lists are built in parallel in linear time. In the main cycle, each position serves as $j$ at most once, so the algorithm performs $O(n)$ $\len$ queries to the iterator. Thus we proved 

\begin{theorem}
There is a linear-time word-RAM algorithm which, given a string $s$ over a general alphabet and a number $k$, builds a palindromic $k$-factorization of $s$ or reports that no such factorization exists.
\end{theorem}

\bibliography{new}

\section*{Appendix}

\begin{proof}[Proof of Lemma~\ref{l:series}]
The lemma appeared in \cite{BKRS17} without a proof (the proof had been put into Appendix, which remained unpublished), so we provide the proof here. It is based on two combinatorial facts about palindromes; these facts appeared, in this or similar form, in~\cite{BKRS17,FiciCo,ISugimotoInenagaBannaiTakeda,KosolobovPalk,RuSh18}.\\[2pt]
\textbf{Fact 1} (\cite[Lemmas 2, 3]{KosolobovPalk}). 
For any palindrome $w$ and any $p \in [1..|w|]$, the following conditions are equivalent: (1)~$p$~is a period of $w$, (2)~there are palindromes $u$, $v$ such that $|uv| = p$ and $w = (uv)^ru$ for some $r \ge 1$, (3)~$w[p{+}1..|w|]$ ($w[1..|w|{-}p]$) is a palindrome.\\
\textbf{Fact 2} (\cite[Lemma 7]{KosolobovPalk}). 
Suppose that $w = (uv)^ru$, where $r \ge 1$, $u$ and $v$ are palindromes, and $|uv|$ is the minimal period of $w$; then, the center of any subpalindrome $x$ of $w$ such that $|x| \ge |uv|{-}1$ coincides with the center of some $u$ or $v$ from the decomposition.

If $p=1$, statement 3) trivially holds, so let $p>1$. By Fact~1, the head $w$ of $U$ has the form $(uv)^qu$ for some $q\ge 1$ and some palindromes $u$ and $v$ such that $|uv| = p$ and the longest proper suffix-palindrome of $w$ is $(uv)^{q-1}u$. Applying Fact~1 iteratively, we finally get $U=\{(uv)^qu,\ldots,(uv)^{q-r+1}u\}$. Note that the only suffix-palindrome of $(uv)^2u$ with length ${\ge}|uv|$ is $uvu$ by Fact~2. Hence by Fact~1 the minimal period of $(uv)^2u$ equals $p$. So either $(uv)^2u\in U$ and thus $U$ satisfies 1) or 2), or $U=\{uvu\}$, also satisfying 2). 
\end{proof}

\begin{proof}[Proof of Lemma~\ref{l:left}]
Let $w=(uv)^ku$ be the head of the $p$-series (see Lemma~\ref{l:series}), $x$ be the center of $w$, and $z(uv)^ku$ be the longest suffix of $s[1..n]$ with period $p$ (for $s$ in Example~\ref{ex:left}, $u=\varepsilon$, $v=aba$, $x=n-5/2$, $z=a$). Then $z$ is a proper suffix of $uv$: if $z$ is longer, then $(uv)^{k+1}u$ is a suffix of $s[1..n]$ and thus $w=(uv)^ku$ is not the head of the $p$-series.

Hence the center $x_1$ of the prefix-palindrome $u$ of $w$ satisfies $\len(x_1)=2|z|+|u|$ (in Example~\ref{ex:left}, $x_1=n-11/2$, $\len(x_1)=|aa|=2$). Note that $|u|=\len(x)\bmod p$ and $x_1=2x-\cntr(|u|)$. Thus, $|z|$ and $\leftside[p]=n{-}\len(x){-}|z|$ are computed in $O(1)$ time.
\end{proof}

\subparagraph*{Series evolution.} Here we provide the proofs concerning the evolution of series described in Section~\ref{s:nlogn}, p.\,\pageref{p:evolution}. Some properties appeared in \cite{KosolobovPalk,BKRS17}. We start with the following lemma.
\begin{lemma} \label{l:uvu}
If $U$ is a $p$-series in $s[1..n]$ and $s[1..n{-}1]$ has no $p$-series, then $U=\{uvu\}$ for some palindromes $u,v\ne\varepsilon$.
\end{lemma}
\begin{proof}
Clearly $p>1$, since every string has a 1-series. Recall that a string $w$ is primitive if $w=v^k$ for an integer $k$ means $k=1$. If $p$ is the minimal period of $w$, every length-$p$ substring of $w$ is primitive.

Let $x$ be the center of $w=\head(p)$ in $s[1..n]$. Then $s[1..n{-}1]$ has a suffix-palindrome $w'$ of length $|w|-2$ centered at $x$. Since $w'$ has period $p$ but $s[1..n-1]$ has no $p$-series, $w'$ has some period $q<p$. If $|w'|\ge p+q-\gcd(p,q)$, then by the Fine--Wilf theorem \cite{FineWilf} $w'$ has the period $\gcd(p,q)$; hence the length-$p$ substrings of $w'$ are not primitive. This leads to a contradiction, since all length-$p$ substrings of $w$ are primitive. Therefore, $|w'|< p+q-\gcd(p,q)$ and thus $|w|<2p$. Hence we are in case 2) of Lemma~\ref{l:series} with $r=1$.
\end{proof}
Let $n$ be the number of the iteration at which a new $p$-series $U$ appears. ``New'' means that $s[1..n]$ has a $p$-series while the strings $s[1..\ell{+}p],\ldots, s[1..n{-}1]$ have no $p$-series, where $\ell=\leftside[p]+1$ is such that $s[\ell..n]$ is the longest $p$-periodic suffix of $s[1..n]$. (In other words, either this is the first ever $p$-series, or the period $p$ died since the previous $p$-series was seen.) By Lemma~\ref{l:uvu}, $U=\{uvu\}$. Let $s[\ell..n]=zuvu$. Since $uvuvu\notin U$, $uvuvu$ is not a suffix of $s[1..n]$; then $z$ is a (possibly empty) suffix of $uv$. The palindrome $zu\lvec{z}$ is a prefix of $zuvu$ and thus has the period $p$. However, $p$ is not the minimum period of $zu\lvec{z}$: otherwise we would have an earlier $p$-series within the same $p$-periodic substring, so $U$ would not be new.

At subsequent iterations, the period $p$ survives while the string which is appended to $s[1..n]$ is a prefix of $(vu)^r$ for some $r$. By Fact~2, all palindromes with minimum period $p$ are centered at the centers of some $u$'s and $v$'s. The distance between the centers of consecutive $u$ and $v$ is $\frac p2$. 

\smallskip
Consider the evolution of $U$. Let $x$ be the center of $uvu=\head(p)$. During next $|z|$ iterations, this center will survive, extending the palindrome to $zuvu\lvec{z}$. At this moment the suffix-palindrome centered at $x+\frac p2$ equals $zu\lvec{z}$ and thus has a period shorter than $p$ and does not belong to the $p$-series. Hence if $U_k$ denotes the $p$-series at $(n{+}k)$th iteration, one has $U_{|z|}=\{zuvu\lvec{z}\}$. Now note that at the next iteration $x$ dies: the letter added on the right extends $\lvec{z}$ to a longer prefix of $vu=\overleftarrow{uv}$, while the letter $s[\leftside[p]]$, preceding $z$, does not extend $z$ to a longer suffix of $uv$. Let $n{+}i$ be the number of this iteration (so $i=|z|+1$). Next we note that the suffix-palindrome centered at $x+\frac p2$ joins the $p$-series at the iteration $n{+}i{+}j$ for some $j\ge 0$. If $j=0$, then $p$-series survives at the $(n{+}i)$th iteration: it loses its only element but simultaneously gets another one. If $j>0$, the $p$-series dies and reappears $j$ iterations letter. Finally observe that $j<p-i$ by Lemma~\ref{l:uvu}, because at the $(n{+}p)$th iteration the suffix-palindrome centered at $x+\frac p2$ is $uvuvu$. Moreover, at the $(n{+}p)$th iteration the palindrome centered at $x{+}p$ joins the $p$-series, so
$U_p=\{uvuvu,uvu\}$. Further, at the $(n{+}p{+}i)$th iteration $x{+}\frac p2$ dies, but the series still contains the palindrome centered at $x{+}p$; palindromes centered at $x{+}\frac{3p}2$ and $x{+}2p$ join the series at the $(n{+}p{+}i{+}j)$th and $(n{+}2p)$th iterations  respectively. One gets $U_{2p}=\{uvuvuvu,uvuvu,uvu\}$, the next $p$ iterations follow the same pattern, and so on.

\begin{proof}[Proof of Lemma~\ref{l:storepre}]
We use dynamic arrays (like \emph{vectors} in C$^{++}$). Such an array has \emph{size} (space in use) and \emph{capacity} (allocated space). When the size increases and reaches the capacity, the latter doubles; this step rebuilds the array to provide a constant-time access to its elements and takes the time linear in its size. Thus, adding an element can be done in amortized $O(1)$ time, an existing element can be modified in $O(1)$ time, an array can be cleared by setting its size to 0, and the total allocated space is proportional to the maximum size reached.

For each array $\pre^j[p]$, we store integers $R_1$ and $R_2$ and dynamic arrays $F_1$ and $F_2$, both initially of size 0 and capacity 1. During its time-to-live, $p$-series is processed in three stages.\\
1. When a $p$-series appears (say, at $n$th iteration), we set $R_1$ to the residue class corresponding to the current iteration (see \eqref{e:pre}) and append $\pre^j[p][R_1]$ to $F_1$ ($R_1=i-1$ in terms of Fig.~\ref{f:evolve}; so $R_1<p-1$ and $R_1=1$ in Fig.~\ref{f:pre},\,\ref{f:evolve}). At the next iteration we add $\pre^j[p][R_1{-}1]$ to $F_1$, and so on until $\pre^j[p][0]$ is added. Detecting stage 1 for a live $p$-series: $\size(F_1)\le R_1<p-1$.\\
2. After adding $\pre^j[p][0]$, there is no $p$-series until the $(n{+}i{+}\ell)$th iteration (see Fig.~\ref{f:evolve}). At that iteration, we set $R_2=i{+}\ell{-}1$ ($R_2=2$ in Fig.~\ref{f:pre}) and start to fill $F_2$ in the same way as $F_1$. During each iteration from $(n{+}p)$ to $(n{+}p{+}i{-}1)$ we add to $F_2$ the minimum of the new value and the corresponding value in $F_1$. Detecting stage 2: $R_1<\min\{\size(F_1), p-1\}$; $(n{+}p)$th iteration is reached when $R_2-\size(F_2)=R_1$. \\
3. At the $(n{+}p{+}i)$th iteration we clear $F_1$, set $R_1$ to $p{-}1$, and start to fill $F_1$ with new values. At the iteration $n{+}2p{+}i{-}1$ the size of $F_1$ reaches $p$, so it contains the values for all residue classes. Starting from the iteration $n{+}2p{+}i$ we fill $F_1$ taking the minimum of the new value and the stored one. Detecting stage 3: $R_1=p-1$; $(n{+}p{+}i)$th iteration: $\size(F_1)=p$.

When we need to extract a particular element $\PRE[p][c]$, we take the minimum of an element of $F_1$ and an element of $F_2$. The numbers $R_1$ and $R_2$ allow one to compute in $O(1)$ time the positions of each element in its array or to establish that the element does not exist. This ensures the constant-time access to the elements. The total size of all arrays cannot exceed the total number of series in all prefixes of the input string, which is $O(n\log n)$ by Lemma~\ref{l:lognseries}. Thus the lemma is proved.
\end{proof}

\begin{proof}[Proof of Lemma~\ref{l:tableop}]
We gathered scattered facts on table operations from \cite{BKRS17} in one place with full proofs.

A chunk $A[1..h]$, $h\le t$, is stored as the tuple $(A[1],b_2,\ldots, b_h)$, where the 2-bit codes $b_i$ encode the difference $\pm 1$ or $0$ with the previous value; the fourth code is used to indicate the end of a chunk with less than $t$ elements. Consider the operations one by one.\\
1) \emph{Increment all elements of a chunk $A$}. It suffices to increment $A[1]$.\\
2) \emph{Extract an element: given a chunk $A$ and a number $\ell<|A|$, return the number $A[\ell]$.}  Create a table containing, for any $h<t$, the sum of values encoded by $h$ 2-bit codes. The table has less than $2^{2t}$ entries, each of them is a $(\log t+O(1))$-bit number, for a total size of $O(n^{1/4}\log n)$. Thus, the sum of all codes in a chunk is a table operation. To get $A[\ell]$, we add its result on $b_1,\ldots,b_\ell$ to $A[1]$.\\
3) \emph{Extract a chunk: given chunks $A,B$ such that $|B[1]-A[|A|]|\le 1$ and numbers $l,d\le t$, return the chunk $C=AB[l..r]$, where  $r=\min\{l{+}d{-}1,|AB|\}$ and $AB$ is the concatenation of $A$ and $B$.} Using the previous operation, extract $A[l]$ and $A[|A|]$; $C$ is encoded by $A[l]$ followed some final codes from $A$, the code of $B[1]-A[|A|]$, and some initial codes of $B$. (If $l{+}d\le|A|$, $B$ is not involved.)\\
4) \emph{Reverse a chunk: given a chunk $A$, return $\lvec{A}$.} Extract $A[|A|]$ and use a table to ``invert'' a sequence of 2-bit codes; the table contains $O(2^{2t})$ entries of size $O(2^{2t})$ each, for $O(n^{1/2})$ total size.\\
5) \emph{Concatenate two chunks: given chunks $A$ and $B$, return a chunk $C=AB$ or, if $|AB|>t$, two chunks $C,D$ such that $CD=AB$ and $|C|=t$.} Extract $A[|A|]$, getting $C$ as $A[1]$ plus codes from $A$, plus the code of $B[1]-A[|A|]$ and some initial codes from $B$. To get $D$ (if exists), extract the corresponding symbol of $B$ and append the rest of codes from $B$ to it.\\
6) \emph{Extend a chunk with dummy values: given a chunk $A$ and numbers $l,r$ such that $|A|+l+r\le t$, return a chunk $C$ of length $|A|{+}l{+}r$ such that $C[l{+}1..l{+}|A|]=A$, $C[i{+}1]=C[i]-1$ for each $i\le l$, and $C[i{+}1]=C[i]+1$ for each $i> l+|A|$.} Set $C[1]=A[1]+l$, append $l$ codes of -1 from the left and $r$ codes of 1 from the right to the codes in $A$.\\
7) \emph{Take the minimum of two chunks: given chunks $A,B$, $|A|=|B|$, return the chunk $C$ of length $|A|$ such that $C[i]=\min\{A[i],B[i]\}$ for each $i$.} If $|A[1]-B[1]|\ge 2t$, then $C$ is the chunk with the smaller number. Otherwise, $C[1]=\min\{A[1],B[1]\}$ and $|A[1]-B[1]|$ is a $(\log t+O(1))$-bit number. Construct a table which contains, for the given difference and two sequences of 2-bit codes, the sequence of 2-bit codes for $C$. The table has $O(2^{4t}\cdot t)$ entries with $O(2^{2t})$-bit values, for the total size of $O(n^{3/4}\log n)$ bits; so this is a table operation. 
\end{proof}

\begin{proof}[Proof of Lemma~\ref{l:tableop2}]
The argument is mostly the same as in the proof of Lemma~\ref{l:tableop}. A double chunk $A[1..h]$, $h\le t$,  is given by the numbers $A^0[1],A^0[2],A^1[1],A^2[2]$, and the 2-bit codes $b_3^1,\ldots,b_h^1,b_3^2,\ldots,b_h^2$. Extracting an element, we use a table of size $O(2^{4h})=O(n^{1/2})$ containing, for each sequence of bit codes, a tuple $(j_0,i_0,\mathit{diff}_0,j_1,i_1,\mathit{diff}_1)$, where $j_0,j_1\in\{0,1\},i_0,i_1\in\{1,2\},\mathit{diff}_0,\mathit{diff}_1=O(h)$. The extracted element is  $(A^{j_0}[i_0]+\mathit{diff}_0,A^{j_1}[i_1]+\mathit{diff}_1)$. As can be easily seen from the proof of Lemma~\ref{l:tableop}, the only other operation that cannot be reduced in $O(1)$ time to the extraction of one or two elements is taking the minimum. For the minimum, we use the following observation: in a smoothed chunk, if a 2-bit code represents a value of $f_j$, it refers to the previous value of $f_j$. So we proceed with three table functions. The first function takes the sequence of all codes from a smoothed chunk and parses it into two sequences corresponding to $f_1$ and $f_2$ respectively; we then reconstruct $f_1$ and $f_2$ as chunks (not \emph{double} chunks but the ones from Lemma~\ref{l:tableop}). The second function is the minimum function from Lemma~\ref{l:tableop}; finally, the third function reconstructs the code sequence of a smoothed chunk from the code sequences of its functions $f_1$ and $f_2$. The tables for the first and the third function require $O(n^{1/2}\log n)$ bits each.
\end{proof}

\end{document}